\documentclass[a4paper,UKenglish,cleveref, autoref, thm-restate]{lipics-v2021}

\title{The Support of Bin Packing is Exponential} 

\author{Klaus Jansen}{Kiel University, Department of Computer Science,  Germany}{kj@informatik.uni-kiel.de}{https://orcid.org/0000-0001-8358-6796}{}
\author{Felix Ohnesorge}{Kiel University, Department of Computer Science,  Germany}{foh@informatik.uni-kiel.de}{https://orcid.org/0009-0003-8023-3380}{}
\author{Lis Pirotton}{Kiel University, Department of Computer Science,  Germany}{lpi@informatik.uni-kiel.de}{https://orcid.org/0009-0001-4984-3696}{}
\author{Malte Tutas}{Kiel University, Department of Computer Science,  Germany}{mtu@informatik.uni-kiel.de}{https://orcid.org/0000-0002-1360-4634}{}

\authorrunning{K. Jansen, L. Pirotton and M. Tutas}

\Copyright{K. Jansen, F. Ohnesorge, L. Pirotton and M. Tutas} 

\ccsdesc[500]{Theory of computation~Integer programming}
\ccsdesc[300]{Theory of computation~Packing and covering problems} 

\keywords{Bin Packing, Integer Programming, Support} 

\category{} 

\relatedversion{} 

\funding{This work was supported by the German Research Foundation (DFG) - Project DFG JA 612/27-1}

\acknowledgements{The idea for this work originated during a visit of Klaus Jansen to the MPI - INF Saarbrücken in 2024.}

\nolinenumbers 

\usepackage{amsmath}
\usepackage{amssymb}	
\usepackage{mathtools}
\usepackage{nicefrac}
\usepackage{xfrac}
\usepackage{xspace}
\usepackage{comment}
\usepackage{dsfont}
\usepackage{marvosym}
\usepackage{colortbl}

\usepackage{graphicx} 
\newtheorem{constraints}{Constraints}
\crefname{constraints}{Constraints}{Constraints}
\crefname{lemma}{Lemma}{Lemmas}
\crefname{observation}{Observation}{Observations}
\usepackage{mathtools}
\usepackage{tikz}
\definecolor{subbi}{rgb}{0.53, 0.81, 1}
\definecolor{color2}{rgb}{0.95, 0.53, 1 }
\definecolor{colora3}{rgb}{1, 0.72, 0.53 }
\usetikzlibrary{shapes.misc}

\tikzset{cross/.style={cross out, draw=black, minimum size=2*(#1-\pgflinewidth), inner sep=0pt, outer sep=0pt},
cross/.default={1pt}}

\usepackage{algorithm}
\usepackage[noend]{algpseudocode}

\begin{document}

\maketitle

\begin{abstract}
      Consider the classical Bin Packing problem with $d$ different item sizes $s_i$ and amounts of items $a_i.$ The support of a Bin Packing solution is the number of differently filled bins. In this work, we show that the lower bound on the support of this problem is $2^{\Omega(d)}$. Our lower bound  matches the upper bound of $2^d$ given by Eisenbrand and Shmonin~[Oper.Research Letters '06] up to a constant factor.

      \noindent This result has direct implications for the time complexity of several Bin Packing algorithms, such as Goemans and Rothvoss [SODA '14], Jansen and Klein [SODA '17] and Jansen and Solis-Oba [IPCO '10].
   
     \noindent To achieve our main result, we develop a technique to aggregate equality constrained ILPs with many constraints into an equivalent ILP with one constraint. Our technique contrasts existing aggregation techniques as we manage to integrate upper bounds on variables into the resulting constraint. We believe this technique can be useful for solving general ILPs or the $d$-dimensional knapsack problem.
\end{abstract}
\clearpage
\section{Introduction}
In their seminal work~\cite{EisenbrandS06}, Eisenbrand and Shmonin inspect integer conic combinations $b\in \mathbb{Z}^d$ of a finite set of integer vectors $X\subset \mathbb{Z}^d.$ They provide upper bounds on the size of the smallest subset $X^*\subseteq X$ such that $b$ is an integer conic combination of elements in $X^*.$ This can be interpreted in the context of integer programming.
Integer programs are composed of a matrix $A\in \mathbb{Z}^{d\times n},$ an (optional) cost vector $c\in \mathbb{Z}^n$, a target vector $b\in \mathbb{Z}^d$ and a solution vector $x\in \mathbb{Z}^n$. The goal is to compute the solution vector satisfying $Ax=b$ while minimizing $c^Tx.$ In this context, the set of integer vectors corresponds to the columns of $A$ and the size of the smallest subset $X^*$ is the number of non-zero entries (the \textit{support}) in the solution vector. More precisely, they show that for any polytope $P\subseteq \mathbb{R}^d$ and any integral vector $x\in \mathbb{Z}^{|P\cap \mathbb{Z}^{d}|}_{\geq 0}$ of multiplicities there exists an $x^*\in \mathbb{Z}^{|P\cap \mathbb{Z}^{d}|}_{\geq 0}$ such that $|\text{supp}(x^*)|\leq 2^d$ and $\sum_{p\in P\cap \mathbb{Z}^d}x_p\cdot p  =\sum_{p\in P\cap \mathbb{Z}^d} x^*_p\cdot p.$ They achieve this through a sophisticated exchange argument, transforming a solution with a large support into one with a bounded support.
They also show a second bound of $O(d\log(d\Delta)),$ where $\Delta=\|A\|_\infty$ is the largest absolute value in $A$, in a similar fashion. These celebrated results have found application in a variety of applications, including classical problems of computer science like the \textsc{Bin Packing} problem (using the first bound)~\cite{GoemansR20,JansenK20,JansenS11} and \textsc{scheduling on identical} and \textsc{uniform machines} (using the second bound)~\cite{Jansen10,JansenKV20}.

\begin{figure}[t]
    \centering
    \begin{tikzpicture}
        \draw (0,3) -- (0,0) -- (1.5,0)--(1.5,3);
        \draw (2,3) -- (2,0) -- (3.5,0)--(3.5,3);
        \draw (4,3) -- (4,0) -- (5.5,0)--(5.5,3);
        \draw (6,3) -- (6,0) -- (7.5,0)--(7.5,3);
        \node at (-0.2,3) {$C$};
        \draw[dashed] (-0.1,3) -- (7.7,3);
        \draw [fill = white!75!black, fill opacity = 0.7] (0,0) rectangle (1.5,1.2);
        \draw [fill = white!95!black, fill opacity = 0.7] (0,1.2) rectangle (1.5,2.2);
        \draw [fill = white!55!black, fill opacity = 0.7] (0,2.2) rectangle (1.5,3);

        \draw [fill = white!75!black, fill opacity = 0.7] (2,0) rectangle (3.5,0.9);
        \draw [fill = white!95!black, fill opacity = 0.7] (2,0.9) rectangle (3.5,2.1);
        \draw [fill = white!55!black, fill opacity = 0.7] (2,2.1) rectangle (3.5,2.9);

        \draw [fill = white!75!black, fill opacity = 0.7] (4,0) rectangle (5.5,1.5);
        \draw [fill = white!55!black, fill opacity = 0.7] (4,1.5) rectangle (5.5,2.7);

        \draw [fill = white!75!black, fill opacity = 0.7] (6,0) rectangle (7.5,0.3);
        \draw [fill = white!95!black, fill opacity = 0.7] (6,0.3) rectangle (7.5,1.3);
        \draw [fill = white!55!black, fill opacity = 0.7] (6,1.3) rectangle (7.5,2.9);

        \draw [fill = white!75!black, fill opacity = 0.7] (8,0) rectangle node [midway] {$s_1$} (9.5,0.2);
        \draw [fill = white!95!black, fill opacity = 0.7] (8,1) rectangle node [midway] {$s_2$} (9.5,1.2);
        \draw [fill = white!55!black, fill opacity = 0.7] (8,2) rectangle node [midway] {$s_3$} (9.5,2.4);

    \end{tikzpicture}
    \caption{An illustration of the \textsc{Bin Packing} problem with three item types, four bins and capacity $C$. Colors indicate item types. Each bin has a unique \textit{configuration}, i.e., combination of items, assigned to it. Only the leftmost bin is completely filled. The sizes of each item type are shown on the right.}
    \label{fig:BPExample}
\end{figure}
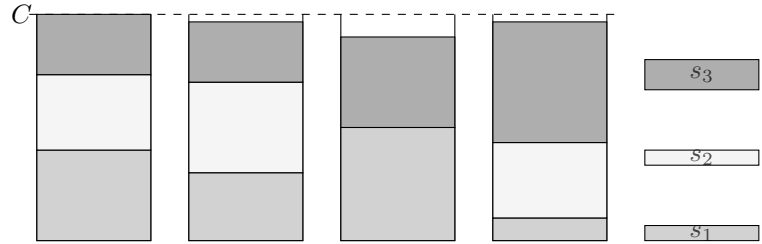
Improvements on the first bound would have direct implications for the time complexity of several algorithms. Consider the high-multiplicity \textsc{Bin Packing} problem: We are given item sizes $s_1,\ldots s_d\in (0,C]$, multiplicities $a_1, \ldots, a_d \in \mathbb{Z}_{\geq 0}$, and a number \(B \in \mathbb{Z}_{\geq0}\).
The task is to decide whether all items can be packed into $B$ bins of size $C$, see \Cref{fig:BPExample} for an illustration.
For this problem, Jansen and Solis-Oba~\cite{JansenS11} give an additive approximation algorithm with guarantee OPT+1. The time complexity of their algorithm is $d^{O(d2^d)}\cdot 2^{O(8^d)}\cdot \text{poly}(enc),$ where $\text{poly}(enc)$ represents a polynomial in the input. In their paper, they ask whether the support bound by Eisenbrand and Shmonin can be improved to be polynomial in $d$. Consequently, this would reduce the dependency on $d$ in the algorithm to be singly-exponential, i.e., result in a running time of $2^{\text{poly}(d)}\cdot\text{poly}(enc).$
\begin{figure}[t]
    \centering
    \begin{tikzpicture}
        \draw[->] (0,0) -- (0,3);
        \draw[->] (0,0) -- (6,0);
        \draw (5,2.5) node[cross=3pt] {};
        \node at (6,2.5) {$t=\begin{pmatrix}
                    5 \\
                    10
                \end{pmatrix}$};
        \node at (-0.3,3) {$x_1$};
        \node at (6,-0.3) {$x_2$};
        \draw[->] (7,2) -- (7.5,2.5);
        \draw[->] (7,1.5) -- (8,1.5);
        \draw[->] (7,0) -- (8.5,1);
        \node at (9.5,2.5) {$p_1=\begin{pmatrix}
                    1 \\1
                \end{pmatrix}$};
        \node at (9.5,1.5) {$p_2=\begin{pmatrix}
                    0 \\2
                \end{pmatrix}$};
        \node at (9.5,0.5) {$p_3=\begin{pmatrix}
                    2 \\3
                \end{pmatrix}$};
        \draw[->] (0,0)  -- node [midway, above=2pt] {$p_3$} (1.5,1);
        \draw[->] (1.5,1)  -- node [midway, yshift=0.1cm, xshift=-0.25 cm] {$p_1$} (2,1.5);
        \draw[->] (2,1.5)  -- node [midway, below=2pt] {$p_2$} (3,1.5);
        \draw[->] (3,1.5)  -- node [midway, yshift=0.1cm, xshift=-0.25 cm] {$p_1$} (3.5,2);
        \draw[->] (3.5,2)  -- node [midway, below=2pt] {$p_2$} (4.5,2);
        \draw[->] (4.5,2)  -- node [midway, yshift=0.1cm, xshift=-0.25 cm] {$p_1$} (5,2.5);
    \end{tikzpicture}
    \caption{An illustration of the \textsc{Cone and Polytope Intersection} problem with three points $p_i\in P$ and the target $\{t\} = Q.$}
    \label{fig:CAPIExample}
\end{figure}
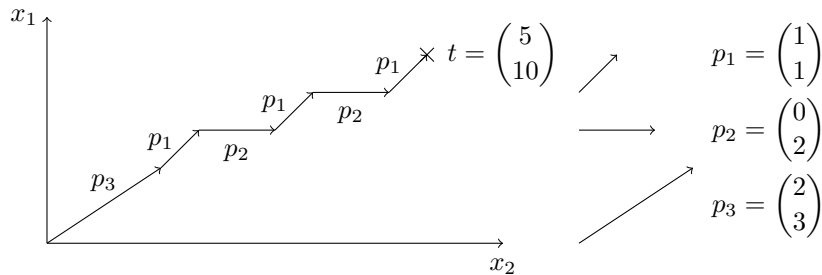

Following this result, Goemans and Rothvoss~\cite{GoemansR20} provide an algorithm running in time $enc(s,a,C)^{2^{O(d)}}.$ Here, $enc(s,a,C)$ denotes the bitsize of encoding the size vector $s$, the multiplicity vector $a$ and the bin capacity $C$ in binary. The time complexity of this algorithm would also be improved by a smaller support bound. For example, a support bound polynomial in $d$ directly improves their running time to $ enc(s,a,C)^{\text{poly}(d)}.$ They achieve this result by solving a more general problem, the \textsc{Cone and Polytope Intersection} (CAPI) problem. Here, we are given two polytopes $P,Q\subseteq \mathbb{R}^d$, where $P$ is bounded. The task is to decide whether there is a point in $Q$ that can be expressed as a non-negative integer combination of integer points in $P$. Goemans and Rothvoss show that this problem captures the essence of high-multiplicity \textsc{Bin Packing}. They reduce \textsc{Bin Packing} to CAPI by setting the integer points in $P=\{\binom{x}{1}\in \mathbb{R}^{d+1}_{\geq 0}| s^Tx\leq C\}$ to every possible configuration, i.e., combinations of items that fit in a single bin. They set $Q=\{\binom{a}{B}\}$ to be the point that corresponds to all items being taken. The final entry in $P$ and $Q$ is used to count the number of used bins $B.$

Jansen and Klein~\cite{JansenK20} extend the work of Goemans and Rothvoss. They give an algorithm for \textsc{Bin Packing} that depends on the number of vertices $|V|$ in the underlying knapsack polytope of a single bin. Their algorithm has a running time of $|V|^{2^{O(d)}} \cdot enc(s,a,C)^{O(1)}$. Thus, it improves upon Goemans and Rothvoss' result when $|V|$ is small. However, it does not improve upon the result by Goemans and Rothvoss in the general sense, as $|V| = \Omega(enc(s,a,C))^d$ as shown by Hayes and Larman~\cite{HayesL83}.

All three of these results could be improved with a polynomial (in $d$) support bound for \textsc{Bin Packing}. However, our result shows that such a bound cannot be achieved, and thus answers the question posed by Jansen and Solis-Oba negatively. Goemans and Rothvoss made first strides in this direction in~\cite{GoemansR20}, where they show that the bound given by Eisenbrand and Shmonin is tight for polytopes, i.e., there is an exponential support when reaching a target with only column vectors or linear combinations of them. They construct a set of column vectors $X$ with dimension $d$ and $|X|=2^{d-1}$ and define the polytope $P=conv(X),$ where $conv(X)$ is the convex hull of $X$. They then define $t=\sum_{x\in X}x$ and show that $t$ can only be reached by taking every element of $X$ exactly once. This shows a lower bound on the support of $2^{d-1}=2^{\Omega(d)}.$ However, their bound does not admit a knapsack constraint and can therefore not be applied to the \textsc{Bin Packing} problem directly. This is because the columns in the \textsc{Bin Packing} ILP originate from a knapsack constraint that describes the feasible packings of a single bin.

With so many different problems gaining important information from support bounds, it should come as no surprise that there are other bounds on the support of integer programs with different parameterizations. Berndt, Brinkop, Jansen, Mnich and Stamm~\cite{BJMS23} show that the support is also bounded by $d\cdot (\log(3\|A\|_1)+\sqrt{\log(\|A\|_1)})$. They complement this result by showing a lower bound on the size of the support of $d\cdot (\lfloor\log(\|A\|_1\rfloor+1).$ Aliev, de Loera, Oertel and O'Neill~\cite{AlievLOO17} give two support bounds for linear diophantine equations of $r(A)+\lfloor\log (H(A))\rfloor$, where $r(A)$ is the column rank of the matrix $A$ and $H(A)$ is the determinant of the span of the $r(A)$-dimensional sublattice of $\mathbb{Z}^n$ formed by all integer points in the linear subspace spanned by the columns of $A$, i.e., $H(A)=\det(\text{span}_{\mathbb{R}}(A)\cap \mathbb{Z}^n).$ Their second bound is $\lfloor2d\log(2\sqrt{d}\Delta)\rfloor$. Aliev, De Loera, Eisenbrand, Oertel and Weismantel~\cite{AlievLEOW18} show the bound of $\lfloor2d\log(2\sqrt{d}\Delta)\rfloor$ for general ILPs.

A main tool we use in the following is \emph{aggregation.} Aggregation is a powerful tool that transforms an ILP with multiple constraints into a single constraint. Research into this technique began in the 70's, resulting in several different methods, see e.g.~\cite{BRADLEY71,BradleyHW74,elmaghraby1970treatment,GloverB95,GloverW72,kendall77,Padberg72,Poirion19,ROSENBERG74}. Elmaghraby and Wig~\cite{elmaghraby1970treatment} iteratively combine two equations. Applying this technique repeatedly yields a single constraint. Another approach is given by Padberg~\cite{Padberg72}. Instead of combining two equations iteratively, he aggregates all equations at the same time.

Our approach is based on the techniques by Padberg~\cite{Padberg72}. We also aggregate all constraints in a single step. However, we manage to include upper bounds on the variables inside the resulting knapsack constraints, where Padberg's method leaves them outside the constraint. We believe that this type of aggregation can find application in different contexts, such as $d$-dimensional knapsack problems or general ILPs with $d$~constraints. Here, the aggregation may yield faster algorithms, especially when one manages to reduce the size of the coefficients. This may be done by using techniques by Frank and Tardos~\cite{FrankT87}, Eisenbrand, Hunkenschröder, Klein, Koutecký Levin and Onn~\cite{EisenbrandHKKLO23} or Jansen, Kahler and Zwanger~\cite{JansenKZ25}.

On the negative side, it is known that aggregation is not feasible for a set of inequalities, as was proven by Chvátal and Hammer~\cite{chvtal1977aggregation}. They provide a simple two-line ILP that cannot be aggregated.

\subsection{Our Contribution}

In this work, we continue research into the support of \textsc{Bin Packing} solutions. 
We show that the support of the high-multiplicity \textsc{Bin Packing} problem is exponential, requiring at least $2^{\Omega(d)}$ many different configurations in an optimum solution. This yields our main result:
\begin{restatable}{theorem}{mainthm}\label{thm:main}
    There exists a high-multiplicity \textsc{Bin Packing} instance of dimension $d$ such that the solution-vector $x$ (a vector of configurations) contains at least $2^{\Omega(d)}$ non-zero entries, i.e., $|\text{supp}(x)| \in 2^{\Omega(d)}.$
\end{restatable}

With this result, we answer an open question posed by Jansen and Solis-Oba in~\cite{JansenS11}. They provide an algorithm for the cutting stock problem. The running time of this algorithm could be improved to be singly-exponential if the support of \textsc{Bin Packing} is not exponential. However, our result answers this question negatively. In a similar vein, our result shows that the algorithm by Goemans and Rothvoss~\cite{GoemansR20} can not be directly improved. The same holds for the algorithm by Jansen and Klein~\cite{JansenK19}.

A major tool to achieve this result is the aggregation of multiple constraints and upper bounds on the variables into a single constraint without any upper bounds. To the best of our knowledge, this type of constraint aggregation has not been done before. Contrasting prior results, our method of aggregation manages to include upper bounds inside the constraint instead of copying them from the original problem.

A preliminary version of this work was published in the proceedings of \emph{European Symposium on Algorithms (ESA) 2025}~\cite{jansen_et_al:LIPIcs.ESA.2025.48}. Compared to that version, this article establishes a slightly improved result while using a simpler proof strategy and less complex constraints. Instead of having $24d+7$ different item sizes\footnote{In the conference version, omitting the upper bound constraints -- when using \Cref{lem:general aggregation} as a black box -- resulted in an incorrect count of item sizes. This issue concerns only the construction, not the correctness of the results: the correct number of distinct item sizes is $24d+7$, rather than $12d+4$.} in the \textsc{Bin Packing} instance, for given dimension $d \in \mathbb{Z}_{\geq 1}$, we now manage to construct a \textsc{Bin Packing} instance with $12d+5$ item sizes where an optimal solution requires $2^d$ distinct configurations instead of the prior $2^d-1$ configurations. Thus, we improve the lower bound from $2^{\frac{d'-7}{24}}-1$ to $2^{\frac{d'-5}{12}}$, where $d'$ is the number of different item types.

\subsection*{Overview}
Here, we give a brief overview of the structure of the document and the structure of our main proof. We begin, in \Cref{sec:aggregation}, by developing a technique with which we can aggregate an ILP containing $d$ equalities into one containing only a single equality. 
Next, in \Cref{subsec:Cons}, we construct an equality ILP with $O(d)$ constraints. Each constraint later defines an item size of the \textsc{Bin Packing} problem. We show that there are $2^{d}$ feasible solutions to this ILP. These solutions $X^*$ are the only configurations of the \textsc{Bin Packing} problem that completely fill a bin. Then, we use our aggregation result in \Cref{sec:aggregation} to transform this ILP into a \textsc{Bin Packing} instance. Finally, we show that when placing all items in the optimal $2^d$ bins, every configuration in $X^*$ has to be taken exactly once. This proves \Cref{thm:main}.

\section{Preliminaries}
For a positive integer $n$, we define $[n] \coloneqq \{1, 2, \dots, n\}$ and $[n]_0 \coloneqq \{0,1, \dots, n\}$.
A vector that contains a certain value in all entries is stylized, e.g.\ the zero vector $0_d$ is written as $\mathbb{0}_d.$
We omit the dimension $d$ if it is apparent from the context.
The support of a $d$-dimensional vector $v$ is the set of its indices with a non-zero entry. It is denoted by $\text{supp}(v) \coloneqq \{i \in [d] \mid v_i \neq 0\}.$ 

\begin{definition}[Bin Packing]
    Given is a set of $d$ different item types with sizes $s_i\in (0,C]$ and amounts $a_i\in\mathbb{Z}_{\ge0}$. Given a number $B\in \mathbb{Z}_{\ge0}$, the goal is to decide whether all items can be assigned to one of $B$ bins, while all bins are filled to at most $C.$
\end{definition}
A \emph{configuration} $x \in \mathbb{Z}_{\geq 0}^d$ is a vector of item multiplicities that fit into a single bin, i.e., $s^Tx\leq C$.

Two powerful tools that we use in this work are integer linear programming and the classical knapsack problem.
\begin{definition}[Integer Linear Programming]
    An integer linear program (ILP) is defined by a matrix $A\in \mathbb{Z}^{d\times n},$ an (optional) cost vector $c\in \mathbb{Z}^n$, an upper bound vector $u \in \mathbb{Z}_{\geq 0}^n$ and a target vector $b\in \mathbb{Z}^d$. The goal is to compute a solution vector $x\in \mathbb{Z}^n_{\geq 0}$ satisfying $Ax=b$ and $x \leq u$ while minimizing $c^Tx.$
\end{definition}

\begin{definition} [Unbounded Knapsack]
    Given $n$ items each defined by a profit $p_1, \ldots, p_n \in \mathbb{R}_{\ge 0}$ and a weight $w_1, \ldots, w_n \in \mathbb{Z}_{\ge1}$, along with a capacity $C \in \mathbb{Z}_{\ge1}$. Find a solution vector $x \in \mathbb{Z}_{\ge0}^n$ such that $\sum_{i=1}^n w_i x_i \leq C$ and $\sum_{i=1}^n p_i x_i$ is maximal. We call $\sum_{i=1}^n w_i x_i = C$ the \textit{equality knapsack constraint} to the corresponding problem.
\end{definition}
Following this definition, we call the set of all solutions fulfilling $x_i \geq 0$, $\forall i\in[n]$ and $\sum_{i=1}^n w_i x_i  \leq C$ the \emph{knapsack polytope}.

\section{Aggregation of an ILP}\label{sec:aggregation}
In this section, we show how to transform an ILP with multiple equality constraints into a single knapsack constraint. The novel aspect of our technique is that we include the external upper bounds of the variables in the aggregation.

Consider an ILP of the form $\min\{c^Tx \mid Ax=b, x \in \mathbb{Z}_{\ge0}^n, x \le u\}$ with $A = (a_{ij})_{i\in [d],j\in[n]} \in \mathbb{Z}^{d\times n}, c \in \mathbb{Z}^n, u \in \mathbb{Z}_{\ge0}^n$ and $b\in \mathbb{Z}^d$ and let $\Delta \coloneqq \|A\|_\infty$ be the largest absolute value in $A$.
For each variable $x_j, j \in [n]$, we define the constraint
\begin{align*}
    x_j + y_j = u_j,
\end{align*}
where $y_j \in \mathbb{Z}_{\geq0}$ is a slack variable. This simulates the upper bound $x_j\le u_j$.
Define $U \coloneqq \sum_{j=1}^n u_j$ as the overall upper bound on the sum of all variables. For this, add the constraint:
\begin{align} \label{eq: upper bound}
    \sum_{j=1}^n (x_j +y_j)+y_{n+1} = U,
\end{align}
with $y_{n+1} \in \mathbb{Z}_{\geq0}$ as a slack variable.
Note that bounding the sum of variables by the sum of their upper bounds does not change the set of solutions.

This results in a new ILP $A' \begin{pmatrix}
        x \\ y
    \end{pmatrix} = (b,\;u,\; U)$, where
\begin{align*}
    A' \coloneqq
    \begin{pmatrix}
        A              & 0_{dn}         & 0 \\
        I_n            & I_n            & 0 \\
        \mathds{1}_{n} & \mathds{1}_{n} & 1
    \end{pmatrix}\in\mathbb{Z}^{(d+n+1)\times (2n+1)}.
\end{align*} 

Next, we use the bound $U$ to define a large integer $M$, which we then use to aggregate the constraints:
\begin{equation}\label{eq:set M}
    \begin{array}{rl}
        M & \coloneqq \Delta \cdot U + \max(\|b\|_\infty, \|u\|_\infty) + \Delta + 2.
    \end{array}
\end{equation}

Multiplying both sides of each constraint with $1,M,M^2,\dots,M^{d+n}$ generates the following ILP where all variables have to be non-negative integers.
\begin{equation} \label{eq:ilp to aggregate}
    \begin{array}{rcl}
        \sum_{j=1}^n a_{1j} x_j                 & =      & b_1                                        \\
        M\sum_{j=1}^n a_{2j} x_j                & =      & Mb_2                                       \\
                                                & \vdots                                              \\
        M^{d-1}\sum_{j=1}^n a_{dj} x_j          & =      & M^{d-1}b_d                                 \\
        M^d (x_1 + y_1)                         & =      & M^d u_1                                    \\
                                                & \vdots                                              \\
        M^{d+n-1} (x_n +y_n)                    & =      & M^{d+n-1} u_n                              \\
        M^{d+n}(\sum_{j=1}^n (x_j+y_j)+y_{n+1}) & =      & M^{d+n} U                                  \\
        x_j                                     & \in    & \mathbb{Z}_{\ge0}\quad \forall j \in [n]   \\
        y_j                                     & \in    & \mathbb{Z}_{\ge0}\quad \forall j \in [n+1]
    \end{array}
\end{equation}

An intuitive view of this aggregation technique is that the resulting single constraint solution is a base $M$ number. Here, the smallest digit represents the solution of the first constraint, the second smallest digit represents the solution of the second constraint and so on. 
By multiplying the upper-bound constraint \eqref{eq: upper bound} with the largest number $M^{d+n}$, we ensure that it may never be broken and keeps the sum of the variables in its range.

Next, we show that this ILP can be aggregated into a single equality knapsack constraint by proving that their sets of solutions are identical.

\begin{restatable}[\Rightscissors]{lemma}{aggregation}
    \label{lem:general aggregation}
    The vector $\mathtt{sol} = (x,\;y)^T \in \mathbb{Z}_{\ge0}^{2n+1}$ is a feasible integer solution to \eqref{eq:ilp to aggregate} if and only if $\mathtt{sol}$ is an integer solution to
    \begin{equation} \label{eq: aggregated}
        \begin{array}{rl}
            \sum_{i=1}^d \big(M^{i-1} \sum_{j=1}^n (a_{ij} x_j)\big)
            &+ \sum_{j=1}^n \big(M^{d+j-1} (x_j + y_j)\big)\\
            &+ M^{d+n}\big(\sum_{j=1}^n (x_j+y_j)+y_{n+1}\big) \\
            = \sum_{i=1}^d (M^{i-1} b_i)
            &+ \sum_{j=1}^n \big(M^{d+j-1} u_j\big)
            + M^{d+n} U.
        \end{array}
    \end{equation}
\end{restatable}
\begin{proof}
    $\Longrightarrow:$ Let $\mathtt{sol} = (x,\;y)^T \in \mathbb{Z}_{\ge0}^{2n+1}$ be a feasible solution to \eqref{eq:ilp to aggregate}.
    For general ILPs that only contain equalities, the sum of all constraints must equal the sum of all entries in the right-hand side. Applied to our ILP $A'$, we have the following:
    If some $x \in \mathbb{Z}_{\ge0}^{2n+1}$ fulfills $A'x = b', b'\in \mathbb{Z}^{d+n+1}$, then $x$ also fulfills $\sum_{i = 0}^{d+n+1} \sum_{j=0}^{2n+1} A'_{ij} x_j = \sum_{i = 0}^{d+n+1} b'_i$.
    In our case this implies that in \eqref{eq:ilp to aggregate}, the sum of the left-hand-sides of all equations is equal to the sum of the right-hand sides because \eqref{eq:ilp to aggregate} only contains equalities. Therefore, $\mathtt{sol}$ is a feasible solution to \eqref{eq: aggregated}.

    $\Longleftarrow:$
    Let $\mathtt{sol} = (x,\;y)^T \in \mathbb{Z}_{\ge0}^{2n+1}$ be a feasible solution to \eqref{eq: aggregated}.
    Then $\mathtt{sol}$ is a feasible solution to
    \begin{equation} \label{eq: equalZero}
        \begin{split}
            \sum_{i=1}^d \big(M^{i-1} \big(\sum_{j=1}^n (a_{ij} x_j) - b_i\big)\big) + \sum_{j=1}^n \big(M^{d+j-1} (x_j + y_j - u_j)\big) &\\ + M^{d+n}\big(\sum_{j=1}^n (x_j+y_j)+y_{n+1} - U\big) &= 0.
        \end{split}
    \end{equation}
    Note that the equation \eqref{eq: equalZero} contains a single summand for each factor $M^k, k\in [d+n]_0.$ Let the \textit{term} $T_\ell$ refer to the summand containing $M^{\ell}$, i.e., we have $T_0=M^{0} \big(\sum_{j=1}^n (a_{1j} x_j) - b_1\big).$ Accordingly, the final term is $T_{d+n}=M^{d+n}\big(\sum_{j=1}^n (x_j+y_j)+y_{n+1} - U\big).$

    We show by induction that each term on the left-hand side of the equation equals 0. For this, we prove for all $i \in \{d+n,d+n-1,\ldots,0\}$ that if $T_i \neq 0$, then $|\sum_{k=1}^{i-1} T_k| < |T_i|$, i.e., no non-zero value can be compensated for by smaller terms. This then leads to a contradiction to \eqref{eq: equalZero}. Thus, we have a contradiction to the assumption that $\mathtt{sol}$ is a feasible solution to \eqref{eq: aggregated}.

    \textit{Base Case:}
    Consider the term of highest power $T_{d+n} \coloneqq M^{d+n}\big(\sum_{j=1}^n (x_j+y_j)+y_{n+1} - U\big)$ and assume $T_{d+n}\neq0$.
    Since $\Delta \ge a_{ij}, i\in[n], j\in [d]$ and $\Delta \ge 1$, we can bound the sum of terms of lower power by:
    \begin{align*}
        \sum_{i=0}^{d+n-1} T_i \le \sum_{i=0}^{d+n-1} M^{i}\big(\Delta \sum_{j=1}^{n}(x_j+y_j) + \max(\|b\|_\infty, \|u\|_\infty)\big).
    \end{align*}
    To show that the sum of those terms $\sum_{i=0}^{d+n-1} T_i $ can never reach an absolute value within the range of $T_{d+n}$, we consider the following two cases:

    \textit{Case 1:} Assume $U \ge \sum_{j=1}^n (x_j+y_j)$.
    With the geometric sum $\sum_{i=0}^{d+n-1} M^{i} = \frac{M^{d+n}-1}{M-1}$, we get
    \begin{align*}
        |\sum_{i=0}^{d+n-1} T_i| & \le |\sum_{i=0}^{d+n-1} M^{i}\big(\Delta \cdot U + \max(\|b\|_\infty, \|u\|_\infty)\big)| \\
                               & \stackrel{\eqref{eq:set M}}{\le} |\sum_{i=0}^{d+n-1} M^{i}(M-1)|                    \\
                               & = \frac{M^{d+n}-1}{M-1}(M-1)                                                        \\
                               & = M^{d+n}-1                                                                         \\
                               & < M^{d+n}.
    \end{align*}
    However, since the variables and coefficients are required to be integer and $M \neq0$, the assumption $T_{d+n}\neq0$ implies that
    \begin{align*}
        \sum_{j=1}^n (x_j+y_j)+y_{n+1} - U \le -1 \qquad \text{or} \qquad \sum_{j=1}^n (x_j+y_j)+y_{n+1} - U \ge 1
    \end{align*}
    which then implies
    \begin{align*}
        T_{d+n} \le -M^{d+n} \qquad \text{or} \qquad T_{d+n} \ge M^{d+n}.
    \end{align*}
    Hence we have that the sum of all terms $T_0$ to $T_{d+n-1}$ can never equalize the term $T_{d+n}$ if $T_{d+n} \neq 0$. Therefore, we have a contradiction to $\eqref{eq: equalZero}$ and thus $T_{d+n} = 0$.

    \textit{Case 2:} Assume $U < \sum_{j=1}^n (x_j+y_j)$ and let $z:= \sum_{j=1}^n (x_j+y_j) - U$. Note that $z$ is a positive integer. With $y_{n+1}\ge0$, the assumption $T_{d+n}\neq0$ implies that
    \begin{align*}
        T_{d+n} \ge M^{d+n}z.
    \end{align*}
    So now we show that the sum of the terms of lower power $\sum_{i=0}^{d+n-1} T_i$ can never reach an absolute value of at least $M^{d+n}z$.
    With \eqref{eq:set M}, we have
    \begin{equation}\label{eq: base case, case 2}
        \begin{array}{crl}
                                & \Delta \cdot U + \max(\|b\|_\infty, \|u\|_\infty) + \Delta + 1           & < M       \\
            \Longrightarrow     & \frac{\Delta \cdot U + \max(\|b\|_\infty, \|u\|_\infty)}{z} + \Delta + 1 & < M       \\
            \Longleftrightarrow & \Delta \cdot U + \max(\|b\|_\infty, \|u\|_\infty) + \Delta z + z         & < Mz      \\
            \Longleftrightarrow & \Delta \cdot U + \max(\|b\|_\infty, \|u\|_\infty) + \Delta z             & < Mz - z. \\
        \end{array}
    \end{equation}
    Using $z +U = \sum_{j=1}^n (x_j+y_j)$, this implies
    \begin{align*}
        |\sum_{i=0}^{d+n-1} T_i| & \le \sum_{i=0}^{d+n-1} M^{i}\big(\Delta (U+z) + \max(\|b\|_\infty, \|u\|_\infty)\big)     \\
                               & = \sum_{i=0}^{d+n-1} M^{i}\big(\Delta \cdot U+ \Delta \cdot z + \max(\|b\|_\infty, \|u\|_\infty)\big) \\
                               & \stackrel{\eqref{eq: base case, case 2}}{<} \sum_{i=0}^{d+n-1} M^{i}(Mz - z)              \\
                               & = \frac{M^{d+n}-1}{M-1}(M - 1)z                                                           \\
                               & = M^{d+n}z-z                                                                              \\
                               & < M^{d+n}z.
    \end{align*}
    From this, we have that the sum of all terms $\sum_{i=0}^{d+n-1} T_i$ can never equalize the term $T_{d+n}$ if $T_{d+n} \neq 0$. Thus, \eqref{eq: equalZero} implies $T_{d+n} = 0$. Since $M^{d+n} \neq 0$, we have $\sum_{j=1}^n (x_j+y_j)+y_{n+1} - U$ and therefore
    \begin{align}\label{eq: sum bound}
        \sum_{j=1}^n (x_j+y_j)+y_{n+1} = U.
    \end{align}

    \textit{Inductive Step:} Let $k \in \{d+n,\dots,1\}$ and assume $T_{k'} = 0$ for all $k' \ge k$. Now, consider term $T_{k-1}$. Either it is of the form $M^{k-1} \big( \sum_{j=1}^n a_{ij}x_j-b_i\big)$ with $i=k$ or it is of the form $M^{k-1} (x_j+y_j-u_j)$ with $j=k-d$. In both cases we have that if $T_{k-1} \neq 0$, then
    \begin{align*}
        T_{k-1} \le -M^{k-1} \qquad \text{or} \qquad T_{k-1} \ge M^{k-1}.
    \end{align*}
    Now, we bound the sum of terms of lower power by
    \begin{equation*}
        \begin{array}{rl}
            |\sum_{i=0}^{k-2} T_i| & \stackrel{\eqref{eq: sum bound}}{\le} |\sum_{i=0}^{k-2} M^{i}\big(\Delta \cdot U + \max(\|b\|_\infty, \|u\|_\infty)\big)| \\
                                   & \stackrel{\eqref{eq:set M}}{\le} |\sum_{i=0}^{k-2} M^{i} (M-1)|                                                     \\
                                   & =\frac{M^{k-1}-1}{M-1}(M-1)                                                                                           \\
                                   & = M^{k-1}-1                                                                                                           \\
                                   & < M^{k-1}.
        \end{array}
    \end{equation*}
    This directly implies that $T_{k-1} = 0$.

    Therefore, we have that $T_i = 0$ for all $i \in [d+n]_0$ which implies that $\mathtt{sol}$ is a feasible solution to \eqref{eq:ilp to aggregate}.
\end{proof}

\section{Bounding the Support for Bin Packing}
Assume we are given a knapsack polytope $P$ of dimension $d+1$. Also assume that there exists a \textsc{Bin Packing} instance $\mathcal{I}$ with just one unique optimal integer solution $x \in \mathbb{Z}_{\ge0}^{d'}$, $d'=O(d)$, where the solution vector $x$ denotes which configurations are chosen in the packing. Additionally, let the $d+1$-dimensional integer points in $P$ be equal to the first $d+1$ dimensions of the configurations. If we now require $k$ different points in $P$ to represent a target or multiplicity vector $a$ as an integer linear combination then $x$ has at least $k$ non-zero entries, which means that we need at least $k$ different configurations in the \textsc{Bin Packing} solution. We construct such an instance where $k$ is exponential in $d'$ and thus show that there exists an exponential lower bound on the support of the \textsc{Bin Packing} problem, i.e., we prove \Cref{thm:main}.

\subsection{Construction of the ILP}\label{subsec:Cons}
In this section, we aim to construct an ILP with equality constraints and upper bounds on the variables, where the first $d+1$ coordinates in its solution vectors form the set
\begin{equation}\label{eq:knapsack points}
    \{(x_1^{\texttt{bin}}, x_2^{\texttt{bin}}, \ldots, x_{d}^{\texttt{bin}}, \gamma^{\sum_{\ell=1}^{d}{2^{d-\ell} \cdot x_{\ell}^{\texttt{bin}}}})\mid x^{\texttt{bin}}_\ell \in \{0,1\} \; \forall \ell \in [d]\}\}.
\end{equation}

Note that we later set $\gamma \coloneqq 4 \cdot 2^d$, then the set of solution vectors will correspond to a certain set of configurations of the \textsc{Bin Packing} instance.
The vectors behave as follows: The last coordinate is of the form $\gamma^i$ for some $i \in \mathbb{Z}_{\ge1}$ and the first $d$ coordinates are the binary encoding of $i.$ In total, there are exactly $2^d$ vectors in the set.

In this section, we construct a set of (in-)equality constraints where, over all feasible solutions, a certain set of variables matches the vectors in \eqref{eq:knapsack points}. In particular, the binary variables $x^\texttt{bin}_1, \dots, x^\texttt{bin}_d$ are going to represent the binary encoding and the variable $r_0$ will be the $(d+1)$-st coordinate, i.e., the set of constraints is only feasible if $r_0 = \gamma^i$ for some $i \in [2^d-1]_0$ and $i = \sum_{\ell=1}^d 2^{d-\ell}\cdot x^\texttt{bin}_\ell$.

After that, in the following sections, we turn the inequalities into equations by introducing slack variables. Then we use the aggregation presented in \Cref{sec:aggregation} to transform these constraints into a single knapsack constraint. We provide a proof that the knapsack constraint is indeed equivalent to the proposed set of constraints (i.e., both allow the same set of solutions) and requires $O(d)$ many variables. With this equivalence, we then have that the first $d+1$ dimensions of the configurations in the \textsc{configIP} form the set \eqref{eq:knapsack points}. Finally, we prove that there exists an exponential lower bound on the support of the \textsc{Bin Packing} problem.

For now, assume that we are given a dimension $d\in\mathbb{Z}_{\ge1}$ and a base $\gamma\in\mathbb{Z}_{\ge2}$. We introduce a set of constraints with $O(d)$ variables which has exactly $2^d$ integer solutions where the variables $x^\texttt{bin}_1,\dots,x^\texttt{bin}_d,$ and $r_0$ behave as desired.

For a brief overview, our constraints behave as follows:
Assume, we are given a binary number $x^\texttt{bin}$ with $d$ digits, where $x^\texttt{bin}_1$ is the most-significant and $x^\texttt{bin}_d$ is the least-significant bit. Leading zeros are allowed.
The variables $r_1,\dots,r_d$ take intermediate values of $r_0$, in detail $r_\ell = \gamma^{\sum_{k=\ell+1}^{d}{2^{d-k} \cdot x_{k}^{\texttt{bin}}}}$ for all $\ell \in [d]_0$.
Note that we can directly fix $r_d = 1$.
In order to correctly set $r_0,\dots,r_{d-1}$, we introduce the auxiliary variables $z_1,\dots,z_d$ which fulfill the property $z_\ell = r_\ell \cdot x_\ell^{\texttt{bin}}.$

The following $4d+1$ constraints in combination with the bounds on the variables, stated below, force the desired set of solutions. Note that all variables should take non-negative values. This is enforced by the ILP condition. Thus, we only state the upper bounds below.

\begin{restatable}{constraints}{cons1}\label{con: original}
    \begin{align}
        r_d                                      & = 1 \label{eq:cons1_1}                                                                                                       \\
        z_\ell                                   & \ge - \gamma^{2^{d-\ell}-1} + \gamma^{2^{d-\ell}-1} x^{\texttt{bin}}_\ell + r_\ell & \forall \ell \in [d] \label{eq:cons1_2} \\
        z_\ell                                   & \le \gamma^{2^{d-\ell}-1} x^{\texttt{bin}}_\ell                                    & \forall \ell \in [d] \label{eq:cons1_3} \\
        z_\ell                                   & \le r_\ell                                                                         & \forall \ell \in [d] \label{eq:cons1_4} \\
        r_{\ell-1} &= (\gamma^{2^{d-\ell}} -1) z_\ell + r_\ell                                                                      & \forall \ell \in [d] \label{eq:cons1_5}
    \end{align}
    Upper Bounds:
    \begin{equation}\label{eq:bounds}
    \begin{array}{rll}
        x^\texttt{bin}_\ell & \leq 1                     & \forall \ell \in [d]   \\
        r_\ell              & \leq \gamma^{2^{d-\ell}-1} & \forall \ell \in [d]_0 \\
        z_\ell              & \leq \gamma^{2^{d-\ell}-1} & \forall \ell \in [d]
    \end{array}
\end{equation}
\end{restatable}

For easier understanding, we present an example.
\begin{example}\label{ex:example}
    Let $\gamma \in \mathbb{Z}_{\ge2}$ and set $d \coloneqq 3$. Then, for $i = 0,\dots,7$, we get the following values for each variable. The highlighted columns correspond to the values in the set \eqref{eq:knapsack points}.\\

    \newcolumntype{C}{>{\centering\arraybackslash}p{0.6cm}}

    \begin{tabular}{|C ||>{\columncolor{subbi!30}} C |>{\columncolor{subbi!30}} C |>{\columncolor{subbi!30}} C ||>{\columncolor{subbi!30}} C | C | C | C || C | C | C |}
        \hline
        $i$ & $x^{\texttt{bin}}_1$ & $x^{\texttt{bin}}_2$ & $x^{\texttt{bin}}_3$ & $r_0$      & $r_1$      & $r_2$      & $r_3$ & $z_1$      & $z_2$      & $z_3$ \\
        \hline
        \hline
        0   & 0                    & 0                    & 0                    & 1 & 1 & 1 & 1     & 0          & 0          & 0     \\
        \hline
        1   & 0                    & 0                    & 1                    & $\gamma^1$ & $\gamma^1$ & $\gamma^1$ & 1     & 0          & 0          & 1     \\
        \hline
        2   & 0                    & 1                    & 0                    & $\gamma^2$ & $\gamma^2$ & 1          & 1     & 0          & 1          & 0     \\
        \hline
        3   & 0                    & 1                    & 1                    & $\gamma^3$ & $\gamma^3$ & $\gamma^1$ & 1     & 0          & $\gamma^1$ & 1     \\
        \hline
        4   & 1                    & 0                    & 0                    & $\gamma^4$ & 1          & 1          & 1     & 1          & 0          & 0     \\
        \hline
        5   & 1                    & 0                    & 1                    & $\gamma^5$ & $\gamma^1$ & $\gamma^1$ & 1     & $\gamma^1$ & 0          & 1     \\
        \hline
        6   & 1                    & 1                    & 0                    & $\gamma^6$ & $\gamma^2$ & 1          & 1     & $\gamma^2$ & 1          & 0     \\
        \hline
        7   & 1                    & 1                    & 1                    & $\gamma^7$ & $\gamma^3$ & $\gamma^1$ & 1     & $\gamma^3$ & $\gamma^1$ & 1     \\
        \hline
    \end{tabular}
\end{example}

The following observations state properties on the feasible assignments of the variables which then lead to the correctness of the constraints.

\begin{observation}\label{obs:properties}
    Let $\gamma\in \mathbb{Z}_{\ge2}$ and $d \in \mathbb{Z}_{\ge1}$.
    Every feasible solution of \Cref{con: original} satisfies:
    \begin{enumerate}
        \item $x^\texttt{bin}_\ell \in \{0,1\}$ for all $\ell \in [d]$,
        \item $z_\ell \in \{0,r_\ell\}$ and $z_\ell = r_\ell x^\texttt{bin}_\ell$ for all $\ell \in [d]$, and
        \item $r_\ell = \gamma^{\sum_{k=\ell+1}^d 2^{d-k}x^\texttt{bin}_k}$ for all $\ell \in [d]_0$.
    \end{enumerate}
\end{observation}
\begin{proof} We prove each statement separately.
    \begin{enumerate}
        \item $x^\texttt{bin}_\ell \in \{0,1\}$ for all $\ell \in [d]$ holds, since we restrict the variables to only take non-negative integer values and we have the upper bound $x^\texttt{bin}_\ell \leq 1$.
        \item Let $\ell \in [d]$. Since $x^\texttt{bin}_\ell\in\{0,1\}$, we may distinguish the following two cases.

              If $x^\texttt{bin}_\ell = 0$, then \eqref{eq:cons1_3} yields $z_\ell \leq 0$. Since all variables are required to be non-negative, we obtain $z_\ell = 0$. Note that the other constraints that restrict $z_\ell$, which are \eqref{eq:cons1_2} and \eqref{eq:cons1_4}, are also fulfilled.

              If $x^\texttt{bin}_\ell = 1$, then \eqref{eq:cons1_2} implies $z_\ell \geq r_\ell$, while \eqref{eq:cons1_4} implies $z_\ell \leq r_\ell$. Hence $z_\ell = r_\ell$. Again note that \eqref{eq:cons1_2} is fulfilled since $r_\ell$ is upper bounded by $\gamma^{2^{d-\ell}-1}$.

              Therefore, for every $\ell \in [d]$, we have $z_\ell \in \{0,r_\ell\}$ and $z_\ell = r_\ell x^\texttt{bin}_\ell$.
        \item We prove this statement by induction on $\ell = d, d-1,\dots,1,0$.

              \emph{Base Case:} Assume $\ell = d$. Then by \eqref{eq:cons1_1}, we have $r_d = 1 = \gamma^0 = \gamma^{\sum_{k=d+1}^d 2^{d-k}x^\texttt{bin}_k}$.

              \emph{Inductive Step:} Let $\ell \in [d]$ and assume $r_\ell = \gamma^{\sum_{k=\ell+1}^d 2^{d-k}x^\texttt{bin}_k}$. We now prove the statement for $\ell-1$. Plugging in $z_\ell = r_\ell x^\texttt{bin}_\ell$ from statement 2 into \eqref{eq:cons1_5} gives $r_{\ell-1} = (\gamma^{2^{d-\ell}} -1) r_\ell x^\texttt{bin}_\ell + r_\ell$.

              If $x^\texttt{bin}_\ell = 0$, then $r_{\ell-1} = r_\ell$.

              If $x^\texttt{bin}_\ell = 1$, then $r_{\ell-1} = \gamma^{2^{d-\ell}}r_\ell$.

              In both cases, we have $r_{\ell-1} = \gamma^{2^{d-\ell} x^\texttt{bin}_\ell}r_\ell$. The induction hypothesis now gives
              \begin{align*}
                  r_{\ell-1} & = \gamma^{2^{d-\ell} x^\texttt{bin}_\ell} \cdot \gamma^{\sum_{k=\ell+1}^d 2^{d-k}x^\texttt{bin}_k} \\
                             & = \gamma^{2^{d-\ell} x^\texttt{bin}_\ell + \sum_{k=\ell+1}^d 2^{d-k}x^\texttt{bin}_k}              \\
                             & = \gamma^{\sum_{k=\ell}^d 2^{d-k}x^\texttt{bin}_k}
              \end{align*}
              which is the desired property.
    \end{enumerate}
\end{proof}

\begin{observation}\label{obs:solutions}
    For all $i \in [2^d-1]_0$ there exists a feasible solution to \Cref{con: original} with $r_0 = \gamma^i$.
\end{observation}
\begin{proof}
    Let $i \in [2^d-1]_0$ and set $x^\texttt{bin}$ such that $i = \sum_{\ell=1}^{d} 2^{d-\ell} x^{\texttt{bin}}_\ell$. Since $i < 2^d$, the number can be encoded in $d$ bits.
    Define the variable $r_d \coloneqq 1$ and recursively set $r_{\ell-1} \coloneqq \gamma^{2^{d-\ell} x^\texttt{bin}_\ell}r_\ell$ for all $\ell \in [d-1]_0$.
    Also, define $z_\ell \coloneqq r_\ell x^\texttt{bin}_\ell$ for all $\ell \in [d]$.

    Note that all variables are non-negative and satisfy the bounds \eqref{eq:bounds}. We now prove that all constraints are also fulfilled.
    \begin{itemize}
        \item[\eqref{eq:cons1_1}] $r_d = 1$ holds by construction.
        \item[\eqref{eq:cons1_2}-\eqref{eq:cons1_4}] Let $\ell\in [d]$. We distinguish the following two cases:\\
              \emph{Case 1:} Assume $x^\texttt{bin}_\ell = 0$. Then $z_\ell = 0$. Thus, $z_\ell = 0 \ge -\gamma^{2^{d-\ell}-1}+r_\ell$, i.e., constraint \eqref{eq:cons1_2}, holds since $r_\ell \leq \gamma^{2^{d-\ell}-1}$.
              Also, $z_\ell = 0 \leq 0 = \gamma^{2^{d-\ell}-1}x^\texttt{bin}_\ell$, i.e., constraint \eqref{eq:cons1_3}, and $z_\ell = 0 \leq r_\ell$, i.e. constraint \eqref{eq:cons1_4}, hold.\\
              \emph{Case 2:} Assume $x^\texttt{bin}_\ell = 1$. Then $z_\ell = r_\ell$. Thus, $z_\ell = r_\ell \ge -\gamma^{2^{d-\ell}-1}+\gamma^{2^{d-\ell}-1}+r_\ell = r_\ell$, i.e., constraint \eqref{eq:cons1_2}, and $z_\ell = r_\ell \leq r_\ell$, i.e. constraint \eqref{eq:cons1_4}, hold.
              Also, $z_\ell = r_\ell \leq \gamma^{2^{d-\ell}-1}$, i.e., constraint \eqref{eq:cons1_3}, holds by the upper bound on $r_\ell$.
        \item[\eqref{eq:cons1_5}] Again let $\ell\in [d]$ and we distinguish the two cases:\\
              \emph{Case 1:} If $x^\texttt{bin}_\ell = 0$, then $z_\ell = 0$ and $(\gamma^{2^{d-\ell}}-1)z_\ell + r_\ell = r_\ell = r_{\ell-1}$ holds.\\
              \emph{Case 2:} If $x^\texttt{bin}_\ell = 1$, then $z_\ell = r_\ell$ and $(\gamma^{2^{d-\ell}}-1)z_\ell + r_\ell = \gamma^{2^{d-\ell}}r_\ell = r_{\ell-1}$ holds.
    \end{itemize}
\end{proof}

\begin{lemma}\label{lem:correctness}
    The set of solutions of \Cref{con: original} consists of exactly $2^d$ distinct vectors where $r_0$ is of the form $\gamma^i$ with $i\in [2^d-1]_0$ and where $x^\texttt{bin}$ is the binary encoding of $i$ (allowing leading zeros), i.e., $i = \sum_{\ell=1}^{d} 2^{d-\ell} x^{\texttt{bin}}_\ell$.
\end{lemma}
\begin{proof}
    \Cref{obs:properties} implies that each feasible solution satisfies $r_0 = \gamma^{\sum_{\ell=1}^d 2^{d-\ell}x^\texttt{bin}_\ell}$ and $x^\texttt{bin}_\ell \in \{0,1\}$ for all $\ell \in [d]$. Thus, $r_0$ takes a value where the exponent of $\gamma$ is integer and takes one of the values in $[2^d-1]_0$. Larger values are not possible as they cannot be encoded by $d$ bits.
    By \Cref{obs:solutions}, we know that all exponents in $[2^d-1]_0$ are feasible.
\end{proof}

Thus, we have shown that the \Cref{con: original} are feasible if and only if variable $r_0$ takes a value of the form $\gamma^i$ for given $\gamma \in \mathbb{Z}_{\ge2}$ and positive integer $i$. Additionally, $x^{\texttt{bin}}$ matches the binary encoding of $i$, i.e., it holds that $i = \sum_{\ell=1}^{d}2^{d-\ell} x^{\texttt{bin}}_\ell$.

\subsection{Aggregation of the Constraints}
Before we can aggregate the constraints into a single knapsack constraint, we need to transform the inequalities into equations such that we achieve an ILP of the form considered in \Cref{sec:aggregation}. We can do that by introducing slack variables which results in the following new set of constraints. The upper bounds on the variables do not change while the bounds on the introduced slack variables can be determined by the coefficients, the variables and the right-hand side of their constraint. 
\begin{constraints}\label{con: equations}
    \begin{align*}
        r_d                                                                      & = 1                                                               \\
        \gamma^{2^{d-\ell}-1} x^{\texttt{bin}}_\ell + r_\ell-z_\ell + y_{1,\ell} & = \gamma^{2^{d-\ell}-1} & \forall \ell \in [d]                    \\
        z_\ell - \gamma^{2^{d-\ell}-1} x^{\texttt{bin}}_\ell + y_{2,\ell}        & = 0                     & \forall \ell \in [d]                    \\
        z_\ell - r_\ell + y_{3,\ell}                                             & = 0                     & \forall \ell \in [d]                    \\
        (\gamma^{2^{d-\ell}} -1) z_\ell + r_\ell - r_{\ell-1}                    & = 0                     & \forall \ell \in [d] \label{eq:cons1_5} \\
        r_0 - y_0                                                                & = 2
    \end{align*}
Upper Bounds:     
    \begin{equation*}
    \begin{array}{rll}
        x^\texttt{bin}_\ell & \leq 1                      & \forall \ell \in [d]   \\
        r_\ell              & \leq \gamma^{2^{d-\ell}-1}  & \forall \ell \in [d]_0 \\
        z_\ell              & \leq \gamma^{2^{d-\ell}-1}  & \forall \ell \in [d]   \\
        y_{1,\ell}          & \leq 2\gamma^{2^{d-\ell}-1} & \forall \ell \in [d]   \\
        y_{2,\ell}          & \leq \gamma^{2^{d-\ell}-1}  & \forall \ell \in [d]   \\
        y_{3,\ell}          & \leq \gamma^{2^{d-\ell}-1}  & \forall \ell \in [d]   \\
        y_0                 & \leq \gamma^{2^d-1}
    \end{array}
    \end{equation*}
\end{constraints}

Note that \Cref{con: equations} consists of $4d+2 $ equality constraints and that the largest absolute value of the coefficients is $\gamma^{2^{d-1}}-1 =: \Delta$.
Now, as done in \Cref{sec:aggregation}, we add the upper bound constraint for each of the $6d+1$ variables. For example, for $x_1^\texttt{bin}$, we get $x_1^\texttt{bin} + y_{x_1^\texttt{bin}} = 1$, where $y_{x_1^\texttt{bin}} \in \mathbb{Z}_{\geq 0}$ is a new variable. 
Additionally, we define $U$ to be the sum of the upper bounds of all variables (including slacks) and introduce the constraint which restricts the sum of all variables by $U$. We denote the sum of all variables by $v$.
Thus, our additional constraint is $v + y_u = U$ with slack variable $y_u \in \mathbb{Z}_{\geq 0}$.
In total, we obtain an ILP with $10d+5$ constraints and $12d+5$ variables.

Now, setting
\begin{align*}
    M & \coloneqq \Delta \cdot U + \max(\|b\|_\infty, \|u\|_\infty) + \Delta + 2 \\
      & = \gamma^{2^{d-1}-1} U + \gamma^{2^d-1} + \gamma^{2^{d-1}-1} + 2
\end{align*}
and multiplying the ILP with the vector $(1,M,M^2,\dots,M^{10d+4})$ results in an ILP of the form \eqref{eq:ilp to aggregate} which is equivalent to the \Cref{con: equations}.

By Lemma 1, this can be aggregated to a single equality constraint. To transform the resulting constraint into a \textsc{Bin Packing} instance, we interpret the variables as items. By grouping the coefficients such that each variable only occurs once, the new coefficient defines the size of the corresponding item.
To complete the transformation, we turn the equation into an unbounded knapsack constraint: We set the lower bound of each variable to 0 and relax the equality into a $\le$ inequality, i.e., the sum of item sizes and multiplicities is at most the right-hand side. While this relaxation expands the set of feasible solutions (configurations), we show in the following section that an optimal \textsc{Bin Packing} solution cannot take any of these new configurations.

\subsection{Main proof}
With this, we are ready to prove our main theorem.

\mainthm*
\begin{proof}
    Let $d\in \mathbb{Z}_{\ge0}$. The \textsc{Bin Packing} instance has dimension $d'$, i.e., the number of different item sizes is $d'=12d+5.$ Denote the capacity of a single bin as $C$ and set the size vector $s$ according to the coefficients in the knapsack constraint. 
    Set $C$ to be the right-hand side of the knapsack constraint. Define the knapsack polytope $P=\{x\in\mathbb{R}^{d'}\mid x_i, s^Tx \leq C\}.$
    Let $X$ be the set $P\cap \mathbb{Z}^{d'}$, i.e., the solutions with total size less than or equal to $C.$ In particular, $X$ is the set of all feasible configurations for bins with capacity $C$. Define $X^*=\{x\in X \mid s^Tx=C\}$ to be the set of all solutions with size exactly $C$, i.e., the set of configurations that completely fill a bin. See \Cref{fig:polytope} for an illustration of $X$ and $X^*$. 

\begin{figure}[t]
    \centering

    \begin{tikzpicture}
        \draw[->] (0,0) -- (0,3);
        \draw[->] (0,0) -- (6,0);
        \node at (-0.2,2.7) {\small $s_1$};
        \node at (5.7,-0.2) {\small $s_2$};
        \draw[opacity=0.3, fill=subbi] (0,0) -- (0,2.6) -- (5.4,0) -- cycle;
        \draw[color=red] (0,2.6)  -- node [midway, above=6pt] {$X^*$}( 5.4,0);
        \foreach \n in {        0.2,0.4,0.6,0.8,1,1.2,1.4,1.6,1.8,2,2.2,2.4,2.6,2.8,3,3.2,3.4,3.6,3.8,4,4.2,4.4,4.6,4.8,5
            }
            {
                \node at (\n,0.2)[circle,fill,inner sep=.7pt,opacity =0.2]{};
            }
        \foreach \n in {        0.2,0.4,0.6,0.8,1,1.2,1.4,1.6,1.8,2,2.2,2.4,2.6,2.8,3,3.2,3.4,3.6,3.8,4,4.2,4.4,4.6
            }
            {
                \node at (\n,0.4)[circle,fill,inner sep=.7pt,opacity =0.2]{};
            }
        \foreach \n in {        0.2,0.4,0.6,0.8,1,1.2,1.4,1.6,1.8,2,2.2,2.4,2.6,2.8,3,3.2,3.4,3.6,3.8,4,4.2
            }
            {
                \node at (\n,0.6)[circle,fill,inner sep=.7pt,opacity =0.2]{};
            }
        \foreach \n in {        0.2,0.4,0.6,0.8,1,1.2,1.4,1.6,1.8,2,2.2,2.4,2.6,2.8,3,3.2,3.4,3.6,3.8
            }
            {
                \node at (\n,0.8)[circle,fill,inner sep=.7pt,opacity =0.2]{};
            }
        \foreach \n in {        0.2,0.4,0.6,0.8,1,1.2,1.4,1.6,1.8,2,2.2,2.4,2.6,2.8,3,3.2,3.4
            }
            {
                \node at (\n,1)[circle,fill,inner sep=.7pt,opacity =0.2]{};
            }
        \foreach \n in {        0.2,0.4,0.6,0.8,1,1.2,1.4,1.6,1.8,2,2.2,2.4,2.6,2.8,3
            }
            {
                \node at (\n,1.2)[circle,fill,inner sep=.7pt,opacity =0.2]{};
            }
        \foreach \n in {        0.2,0.4,0.6,0.8,1,1.2,1.4,1.6,1.8,2,2.2,2.4,2.6
            }
            {
                \node at (\n,1.4)[circle,fill,inner sep=.7pt,opacity =0.2]{};
            }
        \foreach \n in {        0.2,0.4,0.6,0.8,1,1.2,1.4,1.6,1.8,2,2.2
            }
            {
                \node at (\n,1.6)[circle,fill,inner sep=.7pt,opacity =0.2]{};
            }
        \foreach \n in {        0.2,0.4,0.6,0.8,1,1.2,1.4,1.6
            }
            {
                \node at (\n,1.8)[circle,fill,inner sep=.7pt,opacity =0.2]{};
            }
        \foreach \n in {        0.2,0.4,0.6,0.8,1,1.2,1.4
            }
            {
                \node at (\n,2)[circle,fill,inner sep=.7pt,opacity =0.2]{};
            }
        \foreach \n in {        0.2,0.4,0.6,0.8,1
            }
            {
                \node at (\n,2.2)[circle,fill,inner sep=.7pt,opacity =0.2]{};
            }
        \foreach \n in {        0.2,0.4,0.6
            }
            {
                \node at (\n,2.4)[circle,fill,inner sep=.7pt,opacity =0.2]{};
            }
        \foreach \n in {    0.2
            }
            {
                \node at (\n,2.6)[circle,fill,inner sep=.7pt,opacity =0.2]{};
            }
        \node at (1,1){$X$};

    \end{tikzpicture}

    \caption{A 2-dimensional example of the knapsack constraint with item types $s_1,s_2.$ The shaded area denotes the knapsack polytope $P$. The points resemble integer points and show all feasible configurations. Integer points in the shaded region are the configurations in the set $X$. Integer points on the red line form the set $X^*$ of configurations. Only configurations in $X^*$ fill the knapsack constraint with equality. }
    \label{fig:polytope}
\end{figure}
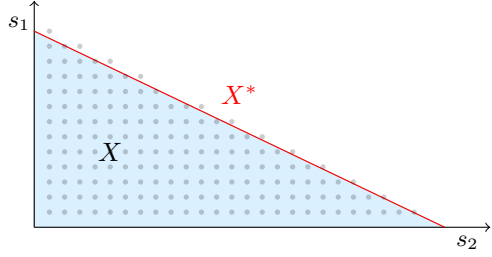

    By the definition of the knapsack constraint and \Cref{lem:correctness}, we know that $\vert X^*\vert =2^{d}=:k,$ as \Cref{con: equations} have exactly one solution for every possible combination of $x^{\texttt{bin}}_\ell, \ell \in [d']$ with $x^{\texttt{bin}}_\ell \in \{0,1\}.$ 
    Recall that the size vector $s$ is given by the coefficients of the aggregated knapsack constraint, while each variable represents the multiplicity how many items of such size are packed in the corresponding configuration. The total amount of items $a$ is now defined as the componentwise sum of the multiplicities of all configurations in $X^*$, i.e., $a \coloneqq \sum_{x \in X^*}x$. Thus, the total size of all items to be placed is $a^T s = kC$.

    Clearly, the optimal solution to this \textsc{Bin Packing} instance uses at least $k$ bins. One possible solution is to take each configuration in $X^*$ exactly once. In the following we show that this is the only optimal solution.
    We cannot consider configurations in $X\setminus X^*$ as they have size strictly less than $C.$ As the total size of the right-hand side is exactly $kC$ and the bins have capacity $C$, taking one of these configurations would imply that we would have to use at least $k+1$ bins which is not optimal.
    Now, set $\gamma \coloneqq 4k$ and consider the coordinate of $r_0$. By the definition of $a$, we have 
    \begin{align}\label{eq:y_r0}
        a_{r_0} = \sum_{x\in X^*} x_{r_0} = \sum_{i \in [k-1]_0} (4k)^i.
    \end{align}

    We now rule out the existence of any solution other than the trivial one. For the sake of contradiction, assume $a_{r_0} = \sum_{i \in [k-1]_0} \lambda_i (4k)^i$ admits a feasible solution $\lambda \in \mathbb{Z}_{\geq 0}^k$ such that $\lambda_i \neq 1$ for some $i \in [k-1]_0$. Let $i'$ be the largest index where $\lambda_{i'} \neq 1$, i.e., $i' \coloneqq \max\{i \mid \lambda_i \neq 1\}$. With \eqref{eq:y_r0}, we have
    \begin{equation*}
        \begin{array}{crl}
                                & \sum_{i \in [k-1]_0} (4k)^i               & = \sum_{i \in [k-1]_0} \lambda_i (4k)^i \\
            \Longleftrightarrow & \sum_{i \in [k-1]_0} (1-\lambda_i) (4k)^i & = 0.
        \end{array}
    \end{equation*}
    Since $\lambda_i = 1$ for all $i > i'$, the equation reduces to
    \begin{align*}
        (1-\lambda_{i'}) (4k)^{i'} + \sum_{i \in [i'-1]_0} (1-\lambda_i) (4k)^i = 0.
    \end{align*}
    Since $\sum_{i \in [k-1]_0} \lambda_i = k$, we have $0 \leq \lambda_i \leq k < 4k$ for all $i \in [k-1]_0$. Thus, no carries are possible and with $k \geq 1$, we get
    \begin{align*}
        \left| \sum_{i \in [i'-1]_0} (1-\lambda_i) (4k)^i \right| 
        \leq \sum_{i \in [i'-1]_0} k (4k)^i 
        = k \cdot \frac{(4k)^{i'}-1}{4k-1} 
        < \frac{k}{4k-1} \cdot (4k)^{i'} 
        < (4k)^{i'}.
    \end{align*}
    However, since $\lambda_{i'}\neq 1$, we also have $|(1-\lambda_{i'})| (4k)^{i'} \geq (4k)^{i'}$. A contradiction.

    Therefore, the only optimal solution $x$ to the \textsc{Bin Packing} instance uses every configuration in $X^*$ exactly once, and therefore has $|\text{supp}(x)|=k=2^{d}.$ Since $d'=12d+5$, we get 
    \begin{align*}
        |\text{supp}(x)|\in 2^{\Omega(d')}.
    \end{align*}
\end{proof}

\clearpage
\bibliography{bib2doi}
\end{document}